\documentclass[reqno,centertags, 12pt, draft]{amsart}

\usepackage{amsmath,amsthm,amscd,amssymb}

\usepackage{latexsym}

\usepackage{mathrsfs}

\newcommand{\bbR}{{\mathbb{R}}}

\newcommand{\bbC}{{\mathbb{C}}}

\newcommand{\no}{\nonumber}

\newcommand{\supp}{\text{\rm{supp}}}

\newcommand{\beq}{\begin{equation}}

\newcommand{\eeq}{\end{equation}}

\newcommand{\ba}{\begin{align}}

\newcommand{\ea}{\end{align}}

\allowdisplaybreaks

\numberwithin{equation}{section}

\newtheorem{theorem}{Theorem}[section]

\newtheorem{lemma}[theorem]{Lemma}

\theoremstyle{definition}

\newtheorem{definition}[theorem]{Definition}

\theoremstyle{remark}

\newtheorem{remark}{Remark}[section]

\date{}
\author{Jonathan Breuer}
\title[Scaling Limits of Jacobi Matrices]{Scaling Limits of Jacobi Matrices and the Christoffel-Darboux Kernel}
\begin{document}
\sloppy
\thanks{Institute of Mathematics, The Hebrew University of Jerusalem, Jerusalem, 91904, Israel.
E-mail: jbreuer@math.huji.ac.il. Supported in part by The Israel Science
Foundation (Grant No. 399/16) and by BSF 2014337}
\maketitle

\begin{abstract}
We study scaling limits of deterministic Jacobi matrices at a fixed point, $x_0$, and their connection to the scaling limits of the Christoffel-Darboux kernel at that point. We show that in the case that the orthogonal polynomials are bounded at $x_0$, a subsequential limit always exists and can be expressed as a canonical system. We further show that under weak conditions on the associated measure, bulk universality of the CD kernel is equivalent to the existence of a limit of a particular explicit form.
\end{abstract}

\section{Introduction}

During the last several years, the subject of scaling limits of random operators has been of considerable utility and interest in the context of random matrix theory \cite{EdelSutt, KRV, RRV, VV1, VV2, ViragRev}. Specifically, in this approach, eigenvalue distribution asymptotics are studied through the identification of a limiting random operator whose eigenvalue distribution corresponds to the limiting distribution of the finite volume eigenvalues. The tridiagonal structure behind the models studied in the references cited above has been central to this approach, which has been applied to study eigenvalue asymptotics of one dimensional Schr\"odinger operators as well \cite{KVV}.

On the deterministic side, the study of asymptotics of eigenvalues of Jacobi matrices is a classical subject going back over seventy years \cite{ET}. This is clear when one recognizes that these eigenvalues are also zeros of orthogonal polynomials \cite{SimonZeros}. Recent renewed interest is due in part to the connection with scaling limits of the Christoffel-Darboux kernel, itself a subject motivated by random matrix theory considerations. Remarkably however, to the best of our knowledge, scaling limits of \emph{deterministic} Jacobi matrices have not been studied. It is our intention in this paper to start filling this gap, and moreover, elucidate the connection of such scaling limits to the scaling limits of the Christoffel-Darboux kernel.

\smallskip

The Christoffel-Darboux (CD) kernel, $K_n^{\mu}(x,y)$, associated with a measure $\mu$ on $\mathbb{R}$, is the integral kernel of the orthogonal projection in $L^2(\mu)$ onto the subspace of polynomials with degree$<n$. Clearly,
\beq \nonumber
K_n^{\mu}(x,y)=\sum_{j=0}^{n-1}p_j^{\mu}(x)p_j^{\mu}(y),
\eeq
where $\left \{p_j^{\mu} \right \}_{j=0}^\infty$ are the orthonormal polynomials associated with $\mu$, i.e., $\textrm{deg}p_j^\mu=j$ and
\beq \nonumber
\int p_j^\mu(x)\overline{p_k^\mu(x)}\textrm{d}\mu(x)=\delta_{j,k}.
\eeq
Note that the $p_j$'s are uniquely determined by this condition, up to the sign of the leading coefficient, which we henceforth take to be positive.

It is a classical result (see, e.g., \cite{deift}) that the polynomials $\{p_j^\mu \}_{j=1}^\infty$ satisfy a recurrence relation, encoded by the \emph{Jacobi matrix} associated with $\mu$
\beq \label{eq:Jacobi}
J^{\mu}=\left(\begin{array}{ccccc} b_1^\mu & a_1^\mu & 0 & \ & \ \\
a_1^{\mu} & b_2^\mu & a_2^\mu & \ddots & \ \\
0 & a_2^\mu & b_3^\mu & a_3^\mu & \ddots \\
\ & \ddots & \ddots & \ddots & \ddots
\end{array} \right)
\eeq
where $a_j^\mu\geq 0$ and $b_j^\mu \in \bbR$. The recurrence relation can be written as a formal `eigenvalue' equation
\beq \label{eq:Recurrence}
J^\mu \left(\begin{array}{c} p_0^\mu (x) \\ p_1^\mu(x) \\ p_2^\mu(x) \\ \vdots \end{array} \right)=x\left(\begin{array}{c} p_0^\mu (x) \\ p_1^\mu(x) \\ p_2^\mu(x) \\ \vdots \end{array} \right)
\eeq
which illustrates the connection between spectral theory and the theory of orthogonal polynomials. Thus, in particular, the measure $\mu$ is the spectral measure of $J^\mu$ and the vector $\delta_1=(1,0,0,\ldots)^{\textrm{t}}$ \cite{deift}. From here on, we shall omit the superscript $\mu$ when there is no risk of confusion.

As the kernel of a projection operator, $K_n$ is a \emph{reproducing kernel} in the sense that
\beq \label{eq:CDReproduce}
K_n(x,y)=\int K_n(x,z)K_n(z,y) \textrm{d}\mu(z).
\eeq
Remarkably, the CD kernel plays an important role in various settings. In the context of random matrix theory, $K_n(x,y)$ plays the role of the correlation kernel for the \emph{orthogonal polynomial ensemble} with measure $\mu$ \cite{Konig}. This is the measure on $\mathbb{R}^n$ given by \beq \no
\textrm{d}\mathbb{P}_n(\lambda_1,\ldots,\lambda_n)=\frac{1}{Z_n}\prod_{1\leq i<j \leq n}|\lambda_i-\lambda_j|^2\textrm{d}\mu(\lambda_1)\ldots\textrm{d}\mu(\lambda_n)
\eeq
where $Z_n$ is a normalization constant. Examples of such ensembles arise naturally in random matrix theory, probability and combinatorics (for a review see \cite{Konig}).

The asymptotic properties (as $n \rightarrow \infty$) of
\beq \label{eq:CDScaling}
\frac{K_n \left(x_0+\frac{a}{n},x_0+\frac{b}{n} \right)}{n}
\eeq
determine the asymptotic properties of $\mathbb{P}_n$ on a scale of size $\frac{1}{n}$ around $x_0$. In this context, the phenomenon of microscopic universality is especially important: in a large number of situations the form of the limit of \eqref{eq:CDScaling} does not depend on the (local or global) properties of $\mu$. In particular, for many classes of measures it has been shown that when $x_0$ is a Lebesgue point of $\mu$ then
\beq \label{eq:BulkUniversality}
\lim_{n \rightarrow \infty} \frac{K_n \left(x_0+\frac{a}{n},x_0+\frac{b}{n} \right)}{n} =\frac{\sin(\pi \rho(x_0)(b-a))}{\pi w(x_0)(b-a)}
\end{equation}
where $\rho$ is the density of the limiting mean empirical measure, (that is, the Radon-Nikodym derivative, w.r.t.\ Lebesgue measure, of the weak limit of $\frac{K_n(x,x)}{n}\textrm{d}\mu(x)$). A very partial list of relevant references is \cite{als,deift,findley,KV,lubinsky2,lubinsky1,lubinskyRev,lubinskyRev2,Simon-ext,totik,totik, XZZ}. 
\begin{remark}
Our scaling \eqref{eq:CDScaling} may seem slightly unusual. Indeed, in many of the works dealing with bulk universality, the limit is considered for 
\beq \no
\frac{K_n \left(x_0+\frac{a}{w(x_0)K_n(x_0,x_0)},x_0+\frac{b}{w(x_0)K_n(x_0,x_0)} \right)}{K_n(x_0,x_0)}
\eeq
which converges (under the appropriate conditions) to $\frac{\sin(\pi(b-a))}{\pi(b-a)}$. However, under the assumptions that $x_0$ is a \emph{strong Lebesgue point} (see Definition \ref{def:SLP} below), which holds for almost every $x_0$ w.r.t.\ the absolutely continuous part of $\mu$, \eqref{eq:BulkUniversality} is equivalent to this more standard formulation of universality together with $\frac{K_n(x_0,x_0)w(x_0)}{n} \rightarrow \rho(x_0)$ (see \cite[Remark 1.2]{bls}). Since \eqref{eq:BulkUniversality} is more convenient for our purposes, we use this formulation.
\end{remark}

Another, more classical, setting in which the asymptotics of \eqref{eq:CDScaling} play an important role is the study of the zeros of the orthonormal polynomial $p_n$. It follows from the Christoffel-Darboux formula (see \eqref{eq:CDFormula} below) that if $p_n(x)=0$ then $K_n(x,y)=0$ iff $y \neq x$ and $p_n(y)=0$. Thus the possible limits of \eqref{eq:CDScaling} also describe the asymptotic behavior of these zeros on a scale of $\frac{1}{n}$ around the point $x_0$. More precisely, if we enumerate the zeros according to their position relative to $x_0$:
\beq \no
\ldots<x_{-2}^{(n)}<x_{-1}^{(n)}<x_0\leq x_0^{(n)}<x_1^{(n)}<\ldots
\eeq
then the asymptotics of $\widetilde{x}_j^{(n)}(x_0)=n(x_{j}^{(n)}-x_0)$ are determined by the asymptotics of \eqref{eq:CDScaling}. In particular, since the zeros of $\sin(\pi x)/(\pi x)$ are the non-zero integers, \eqref{eq:BulkUniversality} implies
\beq \no
\lim_{n \rightarrow \infty} \widetilde{x}_j^{(n)}-\widetilde{x}_{j-1}^{(n)}=\frac{1}{\rho(x_0)}
\eeq
for any fixed $j$ \cite{Freud,LeLu,Simon-CD}.

We note that the $\left \{x_j^{(n)} \right \}$ also have a spectral interpretation: if we let $J^{(n)}$ denote the $n\times n$ truncation of $J$
\beq \label{eq:Jn}
J^{(n)}=\left(\begin{array}{ccccc} b_1 & a_1 & 0 & \ & \ \\
a_1 & b_2 & a_2 & \ddots & \ \\
0 & a_2 & b_3 & a_3 & \ddots \\
\ & \ddots & \ddots & \ddots & \ddots \\
0 & \ldots & \ldots & a_{n-1} & b_n
\end{array} \right)
\eeq
it follows from \eqref{eq:Jacobi} and $p_n \left(x_j^{(n)} \right)=0$ that $\left \{x_j^{(n)} \right \}_j$ are precisely the eigenvalues of $J^{(n)}$.

It is now natural to ask whether there exists a sequence of operators, say $ \left \{\widetilde{J}^{(n)}(x_0)  \right \}_n$ such that the spectrum of $\widetilde{J}^{(n)}(x_0)$ is the set $\left \{\widetilde{x}_j^{(n)}(x_0) \right \}_j$, and that has a limit/limits in an appropriate sense, determining the possible limits of the points $ \left \{\widetilde{x}_j^{(n)}(x_0) \right \}_j$. While such scaled sequences have been identified for various \emph{random} Jacobi matrix ensembles \cite{KRV, KVV, RRV, VV1, VV2, ViragRev}, to the best of our knowledge this problem has not been studied in the deterministic setting. The main point of this paper is the study of the connection between such scaling limits in the bulk (i.e., in the interior of the support of $\mu$) and the asymptotics of \eqref{eq:CDScaling}.

A naive candidate for $\widetilde{J}^{(n)}(x_0)$ would be $n \left(J^{(n)}-x_0I \right)$. However, for fixed $J$, there is no chance that this sequence can have a limit. As we shall show, the right framework for the definition of $\widetilde{J}^{(n)}(x_0)$ is that of the transfer matrices associated with $J$ at $x_0$. We shall construct a sequence of difference operators that is equivalent in a certain sense to the sequence $\{J_n\}_{n}$, and which converges under certain conditions to a limiting differential operator.

While it is not an essential feature in the general approach, additional insight comes from the fact that these operators can be cast in the form of \emph{canonical systems} \cite{debranges,Remling,Romanov}. It is well known that Jacobi matrices may be represented as particular cases of canonical systems and this representation has been shown to be useful in the context of inverse spectral theory (\cite{Remling} treats the analogous continuum case). Moreover, scaling limits of CD kernels have been identified under certain conditions as reproducing kernels of de Branges spaces \cite{LubJFA}. Such spaces are associated with canonical systems \cite{debranges,Remling,Romanov} through the fact that the underlying reproducing kernel can be obtained from the solutions to the corresponding eigenvalue equation. In the random setting, this approach has recently led to a random operator whose eigenvalue distribution corresponds to the asymptotic zero distribution on the critical line of the zeta function \cite{VV2}.

Thus our strategy will be to identify continuous canonical systems as possible limits of the discrete canonical system associated with $J_n$ and $x_0$. Our first main result is a compactness result, asserting that at any point, $x_0$, where the polynomials of the first and second kind (see \eqref{eq:SecondKind} below) are bounded, there always exists a (subsequential) limit. Our second main result establishes an equivalence between \eqref{eq:BulkUniversality} and the characterization of this limit.

This paper is structured as follows. We state our main results and describe the strategy of proof in Section 2. The basic technical convergence results that are presented in Section 3 are used in Section 4 to prove Theorems \ref{thm:Compactness} and \ref{thm:Sine}. Section 5 describes an example illustrating the power of this approach to treat perturbations of strength $O(1/n)$. In the appendix we briefly review some relevant aspects of the theory of canonical systems.


\section{Statement of Main Results}

Fix a measure, $\mu$, with compact support in $\mathbb{R}$, and write
\beq \no
\textrm{d}\mu(t)=w(t)\textrm{d}t+\textrm{d}\mu_{\textrm{sing}}(t),
\eeq
where $\textrm{d}\mu_{\textrm{sing}}$ is the part of $\mu$ that is singular with respect to Lebesgue measure. Let $K_n$ be the associated CD kernel, let $J_n$ be the associated truncated Jacobi matrix as defined in \eqref{eq:Jn} and let $x_0 \in \supp(\mu)$. The key to our approach lies in associating the pair $\left(J_n,x_0 \right)$ with a discrete canonical system that has $K_n$ as its associated reproducing kernel. It is the scaling limit of this canonical system that we shall identify as the scaling limit of $J_n$ at $x_0$. In order to describe this canonical system we first need to describe the transfer matrices associated with $J$ at the point $x_0$.

The $n$-step transfer matrix, $T_n(x)$, at $x$ is the matrix
\beq \label{eq:transfer}
T_n(x)=\left(\begin{array}{cc}  p_n(x) & - q_n(x) \\
a_n p_{n-1}(x) & -a_n q_{n-1}(x) \end{array} \right).
\eeq
Here,
\beq \label{eq:SecondKind}
q_\ell (x)=\int \frac{p_\ell(x)-p_\ell(y)}{x-y}\textrm{d}\mu(y)
\eeq
is the second kind polynomial with respect to the measure $\mu$ (see, e.g., \cite[(1.9)]{bls}). The recurrence relation \eqref{eq:Recurrence} implies that
\beq \label{eq:TransferProduct}
T_n(x)=S_n(x)S_{n-1}(x)S_{n-2}(x)\cdots S_1(x)
\eeq
where the one-step transfer matrix at step $\ell$ is defined by
\beq \label{eq:OneStepTransfer}
S_\ell(x)=\left(\begin{array}{cc} \frac{x-b_\ell}{a_{\ell}} & -\frac{1}{a_{\ell}}  \\
a_{\ell} & 0 \end{array} \right)
\eeq
and we take $a_0 \equiv 1$. Note that
\beq \label{eq:det}
\det S_\ell(x)=1
\eeq
for all $\ell,x$, which implies
\beq \label{eq:det1}
\det T_\ell(x)=1
\eeq
as well.

Now fix $x_0$ and for any $x \in \mathbb{C}$ let
\beq \no
Q_\ell(x)=T_\ell(x_0)^{-1}T_\ell(x).
\eeq
The sequence $\left\{Q_\ell(x) \right\}_\ell$ satisfies
\beq \label{eq:VariationOfParameters}
\begin{split}
Q_{\ell+1}(x)&=T_{\ell}(x_0)^{-1}S_{\ell+1}^{-1}\left(x_0 \right)S_{\ell+1}(x)T_{\ell}(x_0)Q_\ell(x)\\
&=T_{\ell}(x_0)^{-1} \left(\begin{array}{cc} 1 & 0 \\ x_0-x & 1 \end{array} \right) T_{\ell}(x_0)Q_\ell(x) \\
&=\left(\textrm{Id}+(x-x_0)\left( \begin{array}{cc} -p_\ell(x_0)q_\ell(x_0) & q_\ell(x_0)^2 \\ -p_\ell(x_0)^2 & p_\ell(x_0)q_\ell(x_0) \end{array} \right) \right) Q_\ell(x)
\end{split}
\eeq
by \eqref{eq:transfer} and \eqref{eq:det1}. In other words
\beq \label{eq:Difference}
Q_{\ell+1}(x)-Q_\ell(x)=(x-x_0)\left( \begin{array}{cc} -p_\ell(x_0)q_\ell(x_0) & q_\ell(x_0)^2 \\ -p_\ell(x_0)^2 & p_\ell(x_0)q_\ell(x_0) \end{array} \right) Q_\ell(x)
\eeq
or, multiplying both sides by the matrix $\mathcal{J}=\left( \begin{array}{cc} 0 & -1 \\
1 & 0 \end{array} \right)$,
\beq \label{eq:DiscreteCanonical}
\mathcal{J}\left(Q_{\ell+1}(x)-Q_\ell(x)\right) =(x-x_0)H_{x_0,\ell} Q_\ell(x)
\eeq
where
\beq \label{eq:H}
H_{x_0,\ell}=\left(\begin{array}{cc} p_\ell(x_0)^2   & -p_\ell(x_0)q_\ell(x_0)  \\ -p_\ell(x_0)q_\ell(x_0)  & q_\ell(x_0)^2 \end{array} \right)
\eeq
is a nonnegative definite matrix. Equation \eqref{eq:DiscreteCanonical} is a discrete \emph{canonical system} and $Q_\ell(x)$ is a matrix valued solution. We discuss canonical systems in more detail in the appendix.

To understand the relevance of \eqref{eq:DiscreteCanonical} to $K_n(x,y)$ we recall the Christoffel-Darboux formula \cite[Theorem 3.10.4]{SimonSz} which says that for $x\neq y$
\beq \label{eq:CDFormula}
\begin{split}
K_n(x,y)&=\frac{a_n\left(p_n(x)p_{n-1}(y)-p_n(y)p_{n-1}(x) \right)}{x-y} \\
&=\frac{1}{x-y}\det\left(\begin{array}{cc} p_n(x) &  p_n(y) \\ a_n p_{n-1}(x) & a_n p_{n-1}(y) \end{array} \right)
\end{split}
\eeq
or, by \eqref{eq:transfer},
\beq \no
K_n(x,y)=\frac{1}{x-y}\det\left(T_n(x)\left(\begin{array}{c} 1 \\ 0 \end{array} \right),T_n(y)\left(\begin{array}{c} 1 \\ 0 \end{array} \right)  \right).
\eeq

Since the determinant does not change by multiplying both columns by the same matrix of determinant 1, we may multiply both columns by $T_n(x_0)^{-1}$ to get
\beq \label{eq:CDDetFormula}
\begin{split}
K_n(x,y) =\frac{1}{x-y}\det\left(Q_n(x)\left(\begin{array}{c} 1 \\ 0 \end{array} \right),Q_n(y)\left(\begin{array}{c} 1 \\ 0 \end{array} \right)  \right).
\end{split}
\eeq

Now, by taking $x=x_0+\frac{a}{n}$, $y=x_0+\frac{b}{n}$ we see that
\beq \label{eq:ScaledCDKernel}
\begin{split}
&\frac{K_n\left(x_0+\frac{a}{n},x_0+\frac{b}{n} \right)}{n} \\
&\quad =\frac{1}{a-b}\det\left(Q_n(x_0+a/n)\left(\begin{array}{c} 1 \\ 0 \end{array} \right),Q_n(x_0+b/n)\left(\begin{array}{c} 1 \\ 0 \end{array} \right)\right)
\end{split}
\eeq
so that the asymptotics of the right hand side are determined by asymptotics of the solution to the difference equation
\beq \label{eq:ScaledDifference}
\mathcal{J}\left(Q_{\ell+1}-Q_\ell \right) =\frac{a}{n}H_{x_0,\ell} Q_\ell
\eeq
with $Q_0=\textrm{Id}$.
Our results describe conditions under which solutions to \eqref{eq:ScaledDifference} converge to solutions to a corresponding continuum canonical system
\beq \label{eq:ContCanon}
\mathcal{J}Q'(t)=aH(t)Q(t)
\eeq
where $H(t)$ is an appropriately defined limit of the sequence $\{H_{x_0,\ell}\}_{\ell=1}^n$.

\begin{remark}
The process through which \eqref{eq:Difference} is obtained is actually quite standard and shows that any Jacobi matrix may be associated with a discrete canonical system (see \cite{Remling} for the analogue in the continuum and \cite{Romanov} for a different association). In fact, the matrix $Q_\ell(x)$ and the difference equation it satisfies is a compact way of writing the variation of parameters method to express the eigenfunctions of $J$ at $x$ as varying linear combinations of the eigenfunctions at $x_0$ (see e.g.\ \cite{BL} where this matricial form of variation of parameters has been applied in a different context).

An equation similar to \eqref{eq:ScaledDifference} for the free Jacobi matrix (where also a random potential perturbation is taken into account) is also the starting point of the analysis in \cite{KVV}. There convergence is shown to a stochastic differential equation. See also \cite{Last-Simon} for an application of similar ideas to study the spacing of zeros of orthogonal polynomials.
\end{remark}

We are now ready to state our main theorems. Below, $Q_\ell^{(x_0)}(a/n)$ denotes a solution to \eqref{eq:ScaledDifference}.

\begin{theorem} \label{thm:Compactness}
Let $\mu$ be a compactly supported measure on $\mathbb{R}$ and fix $x_0\in\mathbb{R}$. Assume that
\beq \label{eq:BoundedPolys}
\sup_{\ell}(p_\ell(x_0)^2+q_\ell(x_0)^2) <\infty.
\eeq
Then there exists a sequence $\{n_k\}_{k=1}^\infty$ and a measurable function
\beq \no
H:[0,1] \rightarrow M_2(\mathbb{C})
\eeq
such that $H\in L^\infty$, $H(t) \geq 0$ for almost every $t\in[0,1]$ and for any $a \in \mathbb{C}$
$$
Q_{[tn_k]}^{(x_0)}\left(a/n_k\right) \underset{k \rightarrow\infty}{\longrightarrow} \widetilde{Q}_a(t)
$$
uniformly in $t\in [0,1]$, where $\widetilde{Q}_a$ is a solution to the matrix canonical system
\beq \label{eq:MatrixCanonical}
\left( \begin{array}{cc} 0 & -1 \\
1 & 0 \end{array} \right) \widetilde{Q}'(t)=aH(t)\widetilde{Q}(t)  \quad \widetilde{Q}(0)=Id.
\eeq
Moreover, for any $a,b \in \mathbb{C}$,
\beq \no
\begin{split}
&\lim_{k \rightarrow \infty}\frac{1}{n_k}K_{n_j}\left(x_0+\frac{a}{n_k}, x_0+\frac{b}{n_k} \right) \\
& \quad =\frac{1}{a-b} \det \left( \widetilde{Q}_a(1) \left(\begin{array}{c} 1 \\ 0 \end{array} \right), \widetilde{Q}_b(1) \left( \begin{array}{c} 1 \\ 0 \end{array} \right)\right)=K_H(a,b),
\end{split}
\eeq
where $K_H(\overline{a},b)$ is the reproducing kernel for the de Branges space associated with the canonical system \eqref{eq:MatrixCanonical}.
\end{theorem}

\begin{remark}
The condition \eqref{eq:BoundedPolys} says that all solutions of the eigenvalue equation (with eigenvalue $x_0$) of the Jacobi matrix are bounded. It is associated with continuity properties of the measure $\mu$. In particular, it is known that the restriction of $\mu$ to the set of $x$ where \eqref{eq:BoundedPolys} is satisfied, is purely absolutely continuous with respect to Lebesgue measure \cite{simonBounded}. The reverse claim (i.e.\ that for almost all $x$ w.r.t.\ the absolutely continuous part of $\mu$, \eqref{eq:BoundedPolys} holds), known as the Schr\"odinger conjecture, has been disproved by Avila \cite{avila}.
\end{remark}

\begin{remark}
It is interesting to compare this result to \cite[Theorem 1.3]{LubJFA} which says that scaled limits of CD kernels at $x_0$ where $\mu$ is absolutely continuous with a bounded density are reproducing kernels of de Branges spaces that are equal to classical Paley-Wiener spaces. We first point out that the relation of our condition \eqref{eq:BoundedPolys} at $x_0$ to the condition of absolute continuity with a bounded derivative at $\mu$ is not clear at all. While it seems that our condition is somewhat weaker (since it involves only the polynomials \emph{at} $x_0$), it is not known whether this boundedness is indeed implied by sufficiently `nice' behavior of $\mu$ on a neighborhood.

Second, we would like to point out that our conclusion also involves a result for the limits of
\beq \no
\frac{K_{[tn]}\left(x_0+\frac{a}{n},x_0+\frac{b}{n} \right)}{n}
\eeq
for any $t\in [0,1]$ without any extra effort, as these are also determined naturally via are approach.

It would of course be interesting to characterize the possible limiting de Branges spaces/canonical systems in a more concrete way. A natural question is whether every canonical system is a microscopic limit of a (sequence of Jacobi matrix).
\end{remark}

For the next theorem we first need a
\begin{definition} \label{def:SLP}
Let $\textrm{d}\mu(t)=w(t)\textrm{d}t+\textrm{d}\mu_{\textrm{sing}}(t)$ be a measure on $\mathbb{R}$ and let
\beq \label{eq:Stieltjes}
F_\mu(z)=\int \frac{\textrm{d}\mu(t)}{t-z} \quad z \in \mathbb{C}^+
\eeq
be the associated Stieltjes transform. We say that $x$ is a \emph{strong Lebesgue point} of $\mu$ if the following three conditions are satisfied:
\begin{itemize}
\item $F(x+i0)=\lim_{\varepsilon \rightarrow 0+}F(x+i\varepsilon)$ exists and is finite.
\item $\lim_{\varepsilon \rightarrow 0+}\frac{\mu \left(x-\varepsilon, x+\varepsilon \right)}{2\varepsilon}=w(x)>0$.
\item $x$ is a Lebesgue point of $w$.
\end{itemize}
\end{definition}

Note that almost all $x$ with respect to the absolutely continuous part of $\mu$ are strong Lebesgue points.

\begin{theorem} \label{thm:Sine}
Let $x_0$ be a strong Lebesgue point of $\mu$ where we write
\beq \no
\textrm{d}{\mu}(t)={w}(t)\textrm{d}t+\textrm{d}{\mu}_{\textrm{sing}}(t).
\eeq
Let
\beq \no
\textrm{d}\widetilde{\mu}(t)=\widetilde{w}(t)\textrm{d}t+\textrm{d}\widetilde{\mu}_{\textrm{sing}}(t)
\eeq
be the orthogonality measure for the second kind polynomials $\{q_\ell \}_{\ell=0}^\infty$.
Finally, let
\beq \no
H\equiv \left( \begin{array}{cc} \frac{\rho(x_0)}{w(x_0)} & -\textrm{Re}F(x_0+i0)\frac{\rho(x_0)}{w(x_0)} \\
-\textrm{Re}F(x_0+i0)\frac{\rho(x_0)}{w(x_0)} & \frac{\rho(x_0)}{\widetilde{w}(x_0)} \end{array} \right).
\eeq

Then the following are equivalent:

\begin{itemize}
\item[(i)]
$$
\lim_{n \rightarrow \infty} \frac{K_n\left(x_0+\frac{a}{n},x_0+\frac{b}{n} \right)}{n}=\frac{\sin \left(\pi \rho(x_0)(b-a) \right)}{\pi w(x_0) (b-a)}
$$
uniformly for $a,b$ in compact subsets of $\mathbb{C}$, for some $\rho(x_0)>0$.

\item[(ii)] For any $a \in \mathbb{C}$ and $t \in [0,1]$
$$
\lim_{n\rightarrow \infty} Q_{[tn]}^{(x_0)}(a/n)=\widetilde{Q}_a(t)
$$
where $\widetilde{Q}_a(t)$ solves
$$
\left( \begin{array}{cc} 0 & -1 \\
1 & 0 \end{array} \right) \widetilde{Q}'(t)=aH\widetilde{Q}(t)  \quad \widetilde{Q}(0)=Id,
$$
and the convergence is uniform for $a$ in compacts and $t \in [0,1]$.

\end{itemize}
\end{theorem}

\begin{remark}
Note that we do not require boundedness of the Jacobi matrix $x_0$-eigenfunctions, as in \eqref{eq:BoundedPolys}. This is significant since there are models satisfying (i) (and so (ii) and (iii)) of the above theorem but not \eqref{eq:BoundedPolys}. Indeed, it follows from the results of \cite{als} that for any ergodic Jacobi matrix (i) holds for almost every realization at almost every $x_0$ w.r.t.\ the absolutely continuous part of the spectral measure. As shown in \cite{avila}, such models exist where \eqref{eq:BoundedPolys} fails on a set of positive Lebesgue measure.
\end{remark}

Here is the main idea behind the proof of Theorems \ref{thm:Compactness} and \ref{thm:Sine}. Given a sequence of difference equations of the form
\beq \no
X^{(n)}_{j+1}-X^{(n)}_j=\frac{1}{n}A_j^{(n)}X_j^{(n)}
\eeq
on $[0,n]$, rescale $[0,n]$ to $[0,1]$ to obtain the function sequences $\widetilde{X}^{(n)}$ and $\widetilde{A}^{(n)}$. We shall show that, under certain conditions, the convergence
\beq \no
\sup_{0 \leq t \leq 1} \left \|\int_{0}^t \left(\widetilde{A}^{\left(n \right)}(s)-A(s) \right)\textrm{d}s \right \| \rightarrow 0
\eeq
implies the convergence of $\widetilde{X}^{(n)}$ to the solution of
\beq \no
\frac{\textrm{d}X}{\textrm{d}t}=A(t)X(t).
\eeq
Proposition 26 in \cite{KVV} has an analogous random statement, with an additional random factor of the order of $1/\sqrt{n}$. The uniform Lipschitz condition used there translates in our setting to uniform boundedness of the functions $\widetilde{A}^{(n)}$, which is unfortunately too strong for us. We thus need to take some care in our adaptation. This is described in the lemmas in the next section. The proofs of Theorems \ref{thm:Compactness} and \ref{thm:Sine} are completed in Section 4.

\section{Convergence of the Difference Equation}

Assume that for each $n$ we are given a finite sequence of $2 \times 2$ matrices $\left \{A^{(n)}_j \right \}_{j=0}^{n-1}$ and consider the vector difference equation
\beq \label{eq:GeneralDifference}
X^{(n)}_{j+1}-X^{(n)}_j=\frac{1}{n}A_j^{(n)}X_j^{(n)}
\eeq
($X^{(n)}_j \in \bbC^2$) with some prescribed initial condition $X_0^{(n)}=X_0$. We are concerned in this section with conditions ensuring convergence of the solutions to \eqref{eq:GeneralDifference} to the solution of the differential equation
\beq \label{eq:GeneralDifferential}
\frac{\textrm{d}X}{\textrm{d}t}=A(t)X(t)
\eeq
where $A: [0,1] \rightarrow M_2(\bbR)$ is an appropriate limit of $\left \{A^{(n)} \right \}$. It will be useful for us to consider convergence along subsequences.

\begin{lemma} \label{lem:ContinuumConvergence}
For each $n$ define the piecewise constant function \mbox{$\widetilde{A}^{(n)}: [0,1] \rightarrow M_2(\bbR)$} by
\beq \nonumber
\widetilde{A}^{(n)}(t)=A^{(n)}_{[nt]}.
\eeq
Let further
\beq \label{eq:ContintuousXDefinition}
\widetilde{X}^{(n)}(t)=X_{[nt]}^{(n)}
\eeq
where $X^{(n)}$ solves \eqref{eq:GeneralDifference} with initial condition $X_0$.

Assume that there exists a sequence $\{n_k\}_{k=1}^\infty$, $n_k \rightarrow \infty$ such that the following conditions hold:
\beq \label{eq:AverageBound}
\limsup_{k \rightarrow \infty} \frac{1}{n_k} \sum_{j=0}^{n_k} \left \|A_j^{(n_k)} \right \| =C < \infty,
\eeq
\beq \label{eq:LimitBound}
\lim_{k \rightarrow \infty} \frac{1}{n_k} \sup_j \left \|A_j^{(n_k)} \right \|=0,
\eeq
there exists a constant $\widetilde{C}>0$ such that for any $L>0$
\beq \label{eq:AverageDecay}
\limsup_{k \rightarrow \infty} \sup_{t \in [0,1]} \frac{1}{n_k} \sum_{j=\left[\frac{n_k}{L} \cdot [tL] \right]}^{\left[tn_k \right]} \left \| A_j^{\left(n_k \right)}\right \| \leq \frac{\widetilde{C}}{L}.
\end{equation}
Finally, assume that
\beq \label{eq:MatrixConvergence}
\sup_{0 \leq t \leq 1} \left \|\int_{0}^t \left(\widetilde{A}^{\left(n_k \right)}(s)-A(s) \right)\textrm{d}s \right \| \rightarrow 0
\eeq
as $n_k \rightarrow \infty$, where $A: [0,1] \rightarrow M_2(\bbR)$ is an $L^\infty$ matrix valued function.

Then we have that
\beq \label{eq:SolutionConvergence}
\widetilde{X}^{\left(n_k \right)} \rightarrow X
\eeq
uniformly in $[0,1]$ as $k \rightarrow \infty$ where $X:[0,1] \rightarrow \bbC^2$ solves \eqref{eq:GeneralDifferential} with initial condition $X(0)=X_0$.
\end{lemma}

\begin{proof}
First note that $X$ is the unique solution to
\beq \label{eq:IntegralEquation}
X(t)=X_0+\int_0^t A(s)X(s) \textrm{d}s.
\eeq
Further, since $\left \|A(t) \right \|_\infty=K<\infty$ we see that $(t,Y) \mapsto A(t)Y$ is uniformly Lipschitz in $Y \in \bbC^2$. It follows that for any
$ \Phi: [0,1] \rightarrow \bbC^2$ satisfying $\Phi(0)=X_0$
\beq \label{eq:Estimate1}
\left \|X(t)-\Phi(t) \right \|_\infty \leq \sum_{j=0}^\infty  \frac{K^j M\left(\Phi \right)}{j!}=e^K M\left(\Phi \right)
\eeq
where
\beq \nonumber
M \left(\Phi \right)=\sup_t \left \|\Phi(0)+\int_0^t A(s)\Phi(s)\textrm{d}s-\Phi(t) \right \|.
\eeq
This is because the sequence $\Phi^{(j)}$ defined by $\Phi^{(0)}=\Phi$ and
\beq \no
\Phi^{(j)}(t)=\Phi^{(j-1)}(0)+\int_0^t A(s)\Phi^{(j-1)}(s)\textrm{d}s
\eeq
converges uniformly to $X$, so that we may write
\beq \nonumber
X(t)-\Phi(t)=\lim_{j \rightarrow \infty}\Phi^{(j)}(t)-\Phi^{(0)}(t)=\sum_{j=1}^\infty\left(\Phi^{(j)}(t)-\Phi^{(j-1)}(t) \right).
\eeq
Now note that
\beq \nonumber
M\left(\Phi \right)=\left \|\Phi^{(1)}-\Phi^{(0)} \right \|_\infty
\eeq
and that
\beq \nonumber
\Phi^{(j+1)}(t)-\Phi^{(j)}(t)=\int_0^t A(s) \left(\Phi^{(j)}(s)-\Phi^{(j-1)}(s) \right)\textrm{d}s
\eeq
so that \eqref{eq:Estimate1} follows from iteratively estimating the terms in the sum.

Thus, it suffices to show that
\beq \label{eq:BoundConvergence}
\lim_{k \rightarrow \infty} M\left(\widetilde{X}^{\left(n_k \right)} \right)=0.
\eeq
First note that
\beq \label{eq:XnBoundedness}
\begin{split}
\left\| X^{\left(n_k \right)}_j \right\| & =\left\| \left(\textrm{Id}+\frac{A_{j-1}^{\left(n_k \right)}}{n_k} \right)\left(\textrm{Id}+\frac{A_{j-2}^{\left(n_k \right)}}{n_k} \right) \cdots \left(\textrm{Id}+\frac{A_{0}^{\left(n_k \right)}}{n_k} \right)X_0\right \| \\
& \leq e^{\frac{1}{n_k}\sum_{\ell=0}^{j-1} \left \|A_{\ell}^{\left(n_k \right)} \right \|}\left \|X_0 \right \|
\leq e^{C} \left \| X_0 \right \|
\end{split}
\eeq
so that $X^{\left(n_k \right)}_j$ and so also $\widetilde{X}^{\left(n_k \right)}$ is uniformly bounded.

Now note that
\beq \nonumber
\widetilde{X}^{\left(n_k\right)}(t)=X_0+\int_0^{[tn_k]/n_k} \widetilde{A}^{\left(n_k \right)}(s)\widetilde{X}^{\left(n_k \right)}(s) \textrm{d}s
\eeq
so that
\beq \nonumber
\begin{split}
&\left \|\widetilde{X}^{\left(n_k \right)}(t)-X_0-\int_0^t \widetilde{A}^{\left(n_k \right)}(s)\widetilde{X}^{\left(n_k \right)}(s) \textrm{d}s  \right \|_\infty \\
&\quad = \left \| \int_{\frac{[tn_k]}{n_k}}^t \widetilde{A}^{(n_k)}(s) \widetilde{X}^{(n_k)}(s) \textrm{d}s \right \|_\infty \\
& \quad \leq \frac{1}{n_k} \sup_j \left \| A^{(n_k)}_j X^{(n_k)}_j\right \| \rightarrow 0
\end{split}
\eeq
as $k \rightarrow \infty$, by the boundedness of $X^{\left(n_k \right)}_j$ and \eqref{eq:LimitBound}.
Thus
\beq \nonumber
M\left( \widetilde{X}^{\left(n_k \right)} \right)=\left \| \int_0^t \left(\widetilde{A}^{(n_k)}(s)-A(s) \right)\widetilde{X}^{(n_k)}(s) \textrm{d} s \right \|_\infty+o(1).
\eeq
As in the proof of \cite[Proposition 26]{KVV}, we shall show
\begin{equation} \nonumber
\left \| \int_0^t \left[ \left(\widetilde{A}^{(n_k)}-A \right)\widetilde{X}^{(n_k)} \right] (s) \textrm{d} s \right \|_\infty \rightarrow 0
\end{equation}
by subdividing $[0,1]$ into $L<<n_k$ intervals and approximating $\widetilde{X}^{(n_k)}$ by a constant on these intervals.

So fix $L>0$ and define $\widetilde{X}^{(n_k,L)}$ by
\begin{equation} \nonumber
\widetilde{X}^{(n_k,L)}(t)=\widetilde{X}^{(n_k)}\left(\frac{[tL]}{L} \right).
\end{equation}
That is, we divide $[0,1]$ into $L$ equal intervals and set $\widetilde{X}^{(n_k,L)}$ to be constant on each interval with its value equal to the first value of $\widetilde{X}^{(n_k)}$ occurring there.

Now write
\begin{equation} \label{eq:LPartition}
\begin{split}
& \left \| \int_0^t \left[\left(\widetilde{A}^{(n_k)}-A \right)\widetilde{X}^{(n_k)} \right] (s) \textrm{d} s \right \| \\
\quad & \leq \left \| \int_0^t \left[\widetilde{A}^{(n_k)}\left(\widetilde{X}^{(n_k)}-\widetilde{X}^{(n_k,L)} \right) \right](s) \textrm{d} s \right \| \\
&\qquad + \left \| \int_0^t A\left(\widetilde{X}^{(n_k)}-\widetilde{X}^{(n_k,L)} \right) (s) \textrm{d} s \right \| \\
& \qquad +\left \| \int_0^t \left[\left(\widetilde{A}^{(n_k)}-A \right)\widetilde{X}^{(n_k,L)} \right] (s) \textrm{d} s \right \|.
\end{split}
\end{equation}

For the last term write
\begin{equation} \nonumber
\begin{split}
& \int_0^t \left[\left(\widetilde{A}^{(n_k)}-A \right)\widetilde{X}^{(n_k,L)} \right] (s) \textrm{d} s \\
& \quad =\sum_{j=0}^{[tL]-1} \int_{j/L}^{\min\{t,(j+1)/L\}} \left[\left(\widetilde{A}^{(n_k)}-A \right)\widetilde{X}^{(n_k,L)} \right] (s) \textrm{d} s \\
&\quad = \sum_{j=0}^{[tL]-1} \left(\int_{j/L}^{\min\{t,(j+1)/L\}} \left(\widetilde{A}^{(n_k)}-A \right) (s) \textrm{d} s \right)\widetilde{X}^{(n_k,L)} (j/L) \\
\end{split}
\end{equation}
and note that
\begin{equation} \nonumber
\begin{split}
& \left\|\left(\int_{j/L}^{\min\{t,(j+1)/L\}} \left(\widetilde{A}^{(n_k)}-A \right) (s) \textrm{d}s\right)\widetilde{X}^{(n_k,L)} (j/L) \right \| \\
& \quad = \Bigg \| \left( \int_{0}^{\min\{t,(j+1)/L\}} \left(\widetilde{A}^{(n)}-A \right) (s) \textrm{d}s\right) \widetilde{X}^{(n_k,L)} (j/L)  \\
&\qquad -\left(\int_{0}^{j/L} \left(\widetilde{A}^{(n)}-A \right) (s) \textrm{d}s\right) \widetilde{X}^{(n_k,L)} (j/L) \Bigg \| \\
& \quad \leq 2 \sup_t \left \| \int_0^t \left(\widetilde{A}^{(n_k)}(s)-A(s) \right) \textrm{d}s \right \| \left \| \widetilde{X}^{(n_k,L)}(j/L) \right \|
\end{split}
\end{equation}
so that
\begin{equation} \label{eq:LApproximation}
\begin{split}
& \left \| \int_0^t \left[\left(\widetilde{A}^{(n_k)}-A \right)\widetilde{X}^{(n_k,L)} \right] (s) \textrm{d} s \right \|_\infty  \\
&\quad \leq 2L \sup_{j } \left \| \widetilde{X}^{(n_k,L)}(j/L) \right \| \left \|\int_0^t \left(\widetilde{A}^{(n_k)}(s)-A(s) \right) \textrm{d}s \right \|_\infty.
\end{split}
\end{equation}
For any fixed $L$ this is $o(1)$ by \eqref{eq:MatrixConvergence} and \eqref{eq:XnBoundedness}. We shall show that it is possible to choose $L$ so large as to make the other terms on the right hand side of \eqref{eq:LPartition} arbitrarily small. For this, note first that
\begin{equation} \nonumber
\begin{split}
& \left \| \int_0^t \left[\widetilde{A}^{(n_k)}\left(\widetilde{X}^{(n_k)}-\widetilde{X}^{(n_k,L)} \right) \right](s) \textrm{d} s \right \|_\infty \\
&\quad \leq \int_{0}^t \left \|\widetilde{A}^{(n_k)}(s) \right \| \textrm{d}s
\left \| \widetilde{X}^{(n_k)}-\widetilde{X}^{(n_k,L)} \right \|_\infty \\
& \quad \leq \frac{1}{n_k}\sum_{j=0}^{[n_k t]} \left \| A_j^{(n_k)} \right \|
\left \| \widetilde{X}^{(n_k)}-\widetilde{X}^{(n_k,L)} \right \|_\infty \\
& \quad \leq C \left \| \widetilde{X}^{(n_k)}-\widetilde{X}^{(n_k,L)} \right \|_\infty \\
\end{split}
\end{equation}
and similarly
\begin{equation} \nonumber
\begin{split}
& \left \| \int_0^t A\left(\widetilde{X}^{(n_k)}-\widetilde{X}^{(n_k,L)} \right) (s) \textrm{d} s \right \| \\
& \quad \leq K\left \| \widetilde{X}^{(n_k)}-\widetilde{X}^{(n_k,L)} \right \|_\infty.
\end{split}
\end{equation}
Thus we need to show that for any $\varepsilon$ there is $L_0$ such that
\begin{equation} \label{eq:LlargeEnough}
\left \| \widetilde{X}^{(n_k)}-\widetilde{X}^{(n_k,L)} \right \|_\infty < \varepsilon
\end{equation}
for $L > L_0$ for any $n_k$ large enough. Taking $n_k \to \infty$ in \eqref{eq:LApproximation} will then finish the proof.

Write
\beq \no
\widetilde{X}^{(n_k)}(t)=\left(\textrm{Id}+\frac{A^{(n_k)}_{[t n_k]-1}}{n} \right)\cdots\left(\textrm{Id}+\frac{A^{(n_k)}_{\left[\frac{n_k}{L} \cdot [tL] \right]}}{n_k} \right)\widetilde{X}^{(n_k,L)}(t)
\eeq
to see that
\beq \no
\begin{split}
& \left \| \widetilde{X}^{(n_k)}-\widetilde{X}^{(n_k,L)} \right \|_\infty \\
&\quad \leq \sup_t \left \| \left(\left(\textrm{Id}+\frac{A^{(n_k)}_{[tn_k]-1}}{n_k} \right)\cdots  \left(\textrm{Id}+\frac{A^{(n_k)}_{\left[\frac{n_k}{L} \cdot [tL] \right]}}{n_k} \right)-\textrm{Id} \right)\widetilde{X}^{(n_k,L)}(t) \right \| \\
&\quad \leq \sup_t\left \| \left(\textrm{Id}+\frac{A^{(n_k)}_{[tn_k]-1}}{n_k} \right)\cdots \left(\textrm{Id}+\frac{A^{(n_k)}_{\left[\frac{n_k}{L} \cdot [tL] \right]}}{n_k} \right)-\textrm{Id} \right \| \\
& \qquad \times\sup_t \left \| \widetilde{X}^{(n_k,L)}(t) \right \| \\
&\quad \leq e^C \left \|X_0 \right \| \sup_t\left \| \left(\textrm{Id}+\frac{A^{(n_k)}_{[tn_k]-1}}{n_k} \right)\cdots \left(\textrm{Id}+\frac{A^{(n_k)}_{\left[\frac{n_k}{L} \cdot [tL] \right]}}{n_k} \right)-\textrm{Id} \right \|
\end{split}
\eeq
by \eqref{eq:XnBoundedness}. But
\beq \no
\begin{split}
& \sup_t\left \| \left(\textrm{Id}+\frac{A^{(n_k)}_{[tn_k]-1}}{n} \right)\cdots \left(\textrm{Id}+\frac{A^{(n_k)}_{\left[\frac{n_k}{L} \cdot [tL] \right]}}{n_k} \right)-\textrm{Id} \right \| \\
&\quad \leq \left(1+\frac{\left \| A^{(n_k)}_{[tn_k]-1} \right \|}{n_k} \right) \cdots \left(1+\frac{\left \| A^{(n_k)}_{\left[\frac{n_k}{L} \cdot [tL] \right]} \right \|}{n_k} \right)-1 \\
& \quad \leq \exp\left( \frac{1}{n_k}\sum_{j=\left[\frac{n_k}{L} \cdot [tL] \right]}^{[tn_k]} \left \| A^{(n_k)}_j \right \| \right)-1,
\end{split}
\eeq
where the first inequality follows by expanding the product on both sides before applying the triangle inequality. By \eqref{eq:AverageDecay}, for $L>\frac{\widetilde{C}}{1+\varepsilon}$ this is less than $\varepsilon$ for all sufficiently large $n_k$. This finishes the proof.
\end{proof}

Our next two lemmas deal with situations where the conditions of Lemma \ref{lem:ContinuumConvergence} are guaranteed to exist. The first one leads immediately to Theorem \ref{thm:Compactness}. The second one is behind Theorem \ref{thm:Sine}.

\begin{lemma} \label{lem:CompactnessCondition}
With the notation as in the beginning of the section, assume that there exists a constant $C>0$ such that for all $n$
\begin{equation} \label{eq:BoundednessCondition}
\max_j \left \|A_j^{(n)} \right \| \leq C.
\end{equation}
Then there exists an $L^\infty$ function $A: [0,1] \rightarrow M_2(\mathbb{C})$ and a sequence $\{n_k\}_{k=1}^\infty$ with $\lim_{k \rightarrow \infty}n_k=\infty$, such that for any $X_0$, the sequence $\widetilde{X}^{(n_k)}$ defined by \eqref{eq:ContintuousXDefinition} and \eqref{eq:GeneralDifference} with $\widetilde{X}^{(n_k)}(0)=X_0$ converges uniformly to the solution of \eqref{eq:GeneralDifferential} with $X(0)=X_0$.
\end{lemma}

\begin{proof}
Letting
\begin{equation} \nonumber
\widetilde{A}^{(n)}(t)=A^{(n)}_{[tn]}
\end{equation}
the condition \eqref{eq:BoundednessCondition} implies that $\widetilde{A}^{(n)}$ is uniformly bounded in $L^\infty[0,1]$ and so also in $L^2[0,1]$. Thus there is a subsequence, $\widetilde{A}^{(n_k)}$, and a function $A \in L^2\left([0,1], M_2(\mathbb{C})\right)$ such that
\begin{equation} \nonumber
\widetilde{A}^{(n_k)} \underset{k \rightarrow \infty}{\longrightarrow} A
\end{equation}
weakly in $L^2$. Functions in $L^2\left([0,1], M_2(\mathbb{C})\right)$ can be thought of as vector valued square integrable functions, with values in $\mathbb{C}^4$. In particular, the weak convergence implies that for any $0 \leq t\leq 1$
\begin{equation} \nonumber
\int_{0}^t \left(\widetilde{A}^{(n_k)}-A \right)(s) \textrm{d}s= \left(1_{[0,t]},\left(\widetilde{A}^{(n_k)}-A \right)  \right) \rightarrow 0
\end{equation}
as $k \rightarrow \infty$, where $1_{[0,t]}$ has the value $(1,1,1,1)$ on $[0,t]$ and $(0,0,0,0)$ elsewhere. To show uniformity in $t$ assume the convergence is not uniform. By restricting to a subsequence, this implies that for some $\varepsilon>0$ there is a sequence $\{t_k \}_{k=1}^\infty$ such that
\begin{equation} \nonumber
\left \|\int_{0}^{t_k} \left(\widetilde{A}^{(n_k)}-A \right)(s) \textrm{d}s \right \| \geq \varepsilon.
\end{equation}
By restricting to a further subsequence we may assume that $t_k \rightarrow t'$ for some $t' \in [0,1]$. Let $K_1$ be such that for any $k > K_1$
\beq \nonumber
\left \|\int_{0}^{t'} \left(\widetilde{A}^{(n_k)}-A \right)(s) \textrm{d}s \right \| < \frac{\varepsilon}{2}
\eeq
and let $K_2$ be such that for any $k>K_2$
\beq \label{eq:tkt}
|t_k-t'|<\frac{\varepsilon}{4 \left(\|A\|_2+C \right)}.
\eeq
Then for any $k > \max(K_1,K_2)$
\beq \no
\begin{split}
& \left \|\int_{t_k}^{t'} \left(\widetilde{A}^{(n_k)}-A \right)(s) \textrm{d}s  \right \| \\
&\quad = \left \|\int_{0}^{t'} \left(\widetilde{A}^{(n_k)}-A \right)(s) \textrm{d}s-
\int_{0}^{t_k} \left(\widetilde{A}^{(n_k)}-A \right)(s) \textrm{d}s\right \| \\
&\quad  \geq \left \|\int_{0}^{t_k} \left(\widetilde{A}^{(n_k)}-A \right)(s) \textrm{d}s \right \|- \left \| \int_{0}^{t'} \left(\widetilde{A}^{(n_k)}-A \right)(s) \textrm{d}s\right \| \\
&\quad \geq \frac{\varepsilon}{2}.
\end{split}
\eeq
At the same time, however, by \eqref{eq:tkt} and Cauchy-Schwarz
\beq \no
\begin{split}
\left \|\int_{t_k}^{t'} \left(\widetilde{A}^{(n_k)}-A \right)(s) \textrm{d}s  \right \|& = \left \| \left(1_{[t_k,t']}, \left(\widetilde{A}^{(n_k)}-A \right) \right) \right \| \\
&\leq \frac{\varepsilon}{4}.
\end{split}
\eeq
This is absurd. Thus, we see that
\beq \no
\sup_{0 \leq t \leq 1} \left \| \int_{0}^t \left(\widetilde{A}^{(n_k)}(s)-A(s) \right)\textrm{d}s \right\| \rightarrow 0
\eeq
as $k \rightarrow \infty$.
Note in addition that by the Banach-Saks Theorem there is a subsequence $\left\{\widetilde{A}^{\left(n_{k_j} \right)} \right \}_j$ whose Ces\`aro means converge in $L^2$ to $A$. This implies pointwise convergence a.e.\ along a subsequence (of the Ces\`aro means) and thus we see that $\|A \|_\infty \leq C$.

Conditions \eqref{eq:AverageBound} and \eqref{eq:LimitBound} follow immediately from \eqref{eq:BoundednessCondition}. To see that \eqref{eq:AverageDecay} follows as well, simply note that for any $t$ and $L$ the interval $\left[\frac{n_k}{L}\left[tL \right] \right] \leq j \leq [tn_k]$ has at most $(\frac{n_k}{L}+1)$ terms. Thus all the conditions of Lemma \ref{lem:ContinuumConvergence} are satisfied. Its conclusion finishes the proof.
\end{proof}

\begin{lemma} \label{lem:NiceConvergence}
Let $\{\alpha_j\}_{j=0}^\infty$ and $\{\beta_j\}_{j=0}^\infty$ be two sequences of real numbers and for any $j$, let
\beq \label{eq:AjDefinition}
A_j=\left(\begin{array}{cc} \alpha_j^2 & -\alpha_j \beta_j \\
-\alpha_j \beta_j & \beta_j^2 \end{array} \right)
\eeq
Assume that
\beq \label{eq:NiceConvergence}
\lim_{n \rightarrow \infty}\frac{1}{n}\sum_{j=0}^{n-1} A_j=H\in M_2(\mathbb{C}).
\eeq
Then, with $A_j^{(n)}=A_j$ and $A(s) \equiv H$, we have that for any $X_0$, the sequence $\widetilde{X}^{(n)}$ defined by \eqref{eq:ContintuousXDefinition} and \eqref{eq:GeneralDifference} with $\widetilde{X}^{(n)}(0)=X_0$ converges uniformly to the solution of \eqref{eq:GeneralDifferential} with $X(0)=X_0$.
\end{lemma}

\begin{proof}
We need to show that the conditions of Lemma \ref{lem:ContinuumConvergence} are satisfied for the sequence $n_k=k$.
Note first that since
\beq \nonumber
\lim_{n \rightarrow \infty}\frac{1}{n}\sum_{j=0}^{n-1}\alpha_j^2=H_{11} \geq 0
\eeq
and
\beq \nonumber
\lim_{n \rightarrow \infty}\frac{1}{n}\sum_{j=0}^{n-1}\beta_j^2=H_{22} \geq 0
\eeq
we have that
\beq \nonumber
\limsup_{n \rightarrow \infty}\frac{1}{n}\sum_{j=0}^{n-1}|\alpha_j \beta_j| \leq \limsup_{n \rightarrow \infty} \sqrt{\sum_{j=0}^{n-1}\frac{\alpha_j^2}{n}} \sqrt{ \sum_{j=0}^{n-1}\frac{\beta_j^2}{n}} \leq \sqrt{H_{11}H_{22}}
\eeq
so that
\beq \no
\limsup_{n \rightarrow \infty}\frac{1}{n} \sum_{j=0}^{n-1}\left \|A_j \right \| <\infty
\eeq
which is \eqref{eq:AverageBound}.

To show \eqref{eq:LimitBound} it suffices to show that
\beq \label{eq:LimitBound1}
\frac{\left \|A_n \right \|}{n} \rightarrow 0
\eeq
since for any $j \leq n$, $\frac{ \left \|A_j \right \|}{n} \leq \frac{ \left \| A_j \right \|}{j}$, so if \eqref{eq:LimitBound1} holds and there exist sequences $j_\ell \leq n_\ell$ with
$\frac{\left \| A_{j_\ell} \right \|}{n_\ell}> \varepsilon>0$, then $j_\ell$ would have to be bounded which is absurd.
But \eqref{eq:LimitBound1} is clear since
\beq \nonumber
\left \|\frac{1}{n+1}\sum_{j=0}^{n-1}A_j-H \right \| \leq \frac{n}{n+1}\left \|\frac{1}{n}\sum_{j=0}^{n-1} A_j -H  \right\|+\frac{1}{n+1}\| H\|\rightarrow 0
\eeq
as $n \rightarrow \infty$, which implies
\beq \no
\frac{1}{n+1}\|A_n \|=\left \|\frac{1}{n+1}\sum_{j=0}^n A_j-\frac{1}{n+1}\sum_{j=0}^{n-1}A_j \right \| \rightarrow 0.
\eeq

We next prove \eqref{eq:MatrixConvergence}. Note that for any $n$ and $\frac{j}{n} \leq t < \frac{j+1}{n}$
\beq \no
\widetilde{A}^{(n)}(t)=A_j
\eeq
so showing \eqref{eq:MatrixConvergence} is equivalent to showing that
\beq \label{eq:MatrixConvergenceDiscrete}
\lim_{n \rightarrow \infty} \sup_{\ell < n}  \frac{1}{n}\left \| \sum_{j=0}^\ell \left(A_j-H \right) \right \|=0.
\eeq
To show this, fix $\varepsilon>0$ and let $M$ be such that for any $\ell \geq M$,
\beq \no
\frac{1}{\ell}\left \| \sum_{j=0}^\ell \left(A_j-H \right) \right \|<\varepsilon.
\eeq
Now, by \eqref{eq:LimitBound} there exists $N \geq M$ so that for any $n \geq N$
\beq \no
\frac{1}{n} \sum_{j=0}^M \left \| A_j-H \right \| < \varepsilon
\eeq
and thus we see that for any $n \geq N$, for $\ell \leq M$
\beq \label{eq:MatrixConvergenceD1}
\frac{1}{n} \left \| \sum_{j=0}^\ell \left(A_j-H \right) \right \| \leq \frac{1}{n} \sum_{j=0}^M \left \| A_j-H \right \| < \varepsilon
\eeq
and for $M < \ell \leq n$
\beq \label{eq:MatrixConvergenceD2}
\frac{1}{n} \left \| \sum_{j=0}^\ell \left(A_j-H \right) \right \| \leq \frac{1}{\ell} \left \| \sum_{j=0}^\ell \left(A_j-H \right) \right \| <\varepsilon.
\eeq
\eqref{eq:MatrixConvergenceD1} and \eqref{eq:MatrixConvergenceD2} imply \eqref{eq:MatrixConvergenceDiscrete}.

Finally, to show \eqref{eq:AverageDecay} note first that this is equivalent to showing that for any fixed $L$
\beq \label{eq:AverageDecayDiscrete}
\lim_{n \rightarrow \infty}\frac{1}{n}\sum_{j=\left[nk/L \right]}^{\left[n(k+1)/L \right]} \| A_j \| \leq \frac{C}{L}
\eeq
for some $C>0$, uniformly in $k<L$ (write $k=[tL]$). To show this, note first that
\beq \no
\begin{split}
& \frac{1}{n} \left \| \sum_{j=\left[nk/L \right]}^{\left[n(k+1)/L \right]} \left(A_j-H \right) \right \| \\
& \quad \leq
\frac{1}{n} \left \| \sum_{j=0}^{\left[n(k+1)/L \right]} \left(A_j-H \right) \right \| +\frac{1}{n} \left \| \sum_{j=0}^{\left[nk/L \right]} \left(A_j-H \right) \right \| \rightarrow 0
\end{split}
\eeq
as $n \rightarrow \infty$, uniformly in $k$ by \eqref{eq:MatrixConvergenceDiscrete}. Now, clearly (and also uniformly in $k<L$)
\beq \no
\lim_{n \rightarrow \infty}\frac{1}{n} \sum_{j=\left[nk/L \right]}^{\left[n(k+1)/L \right]} H =\frac{H}{L}
\eeq
so we see that
\beq \no
\lim_{n \rightarrow \infty}\frac{1}{n} \sum_{j=\left[nk/L \right]}^{\left[n(k+1)/L \right]} \alpha_j^2=\frac{H_{11}}{L}
\eeq
and similarly
\beq \no
\lim_{n \rightarrow \infty}\frac{1}{n} \sum_{j=\left[nk/L \right]}^{\left[n(k+1)/L \right]} \beta_j^2=\frac{H_{22}}{L},
\eeq
and both limits are uniform in $k$. Since the off diagonal terms of $A_j$ are dominated by the diagonal terms, we see, using Cauchy-Schwarz (as in the first step of the proof), that
\eqref{eq:AverageDecayDiscrete} holds with $C$ that depends on $\| H \|$.

\end{proof}

\section{Proof of Theorems \ref{thm:Compactness} and \ref{thm:Sine}}

\begin{proof}[Proof of Theorem \ref{thm:Compactness}]
Fix $a\in\mathbb{C}$. We apply Lemma \ref{lem:CompactnessCondition} with
\beq \no
X_j^{(n)}=\mathcal{J}Q_j(a/n)
\eeq
and
\beq \no
A_j^{(n)}=aH_{x_0,j}=a\left(\begin{array}{cc} p_j(x_0)^2   & -p_j(x_0)q_j(x_0)  \\ -p_j(x_0)q_j(x_0)  & q_j(x_0)^2 \end{array} \right).
\eeq
Note that
\beq \no
\mathcal{J}Q_{[tn]}(a/n)=\widetilde{X}^{(n)}(t).
\eeq

By \eqref{eq:BoundedPolys}, we see that \eqref{eq:BoundednessCondition} holds and so there exists a sequence $n_k \rightarrow \infty$ and a bounded measurable function $H_a$ such that
\beq \no
\lim_{k\rightarrow\infty}\sup_{0\leq t\leq 1} \left \| \int_0^t \left(aH_{x_0,[n_ks]}-H_a(s) \right)\textrm{d}s \right \|=0
\eeq
and
\beq \label{eq:ConvergenceUniform1}
Q_{[tn_k]}(a/n_k) \underset{k \rightarrow \infty}{\longrightarrow} \widetilde{Q}_a(t)
\eeq
uniformly in $t\in [0,1]$, where
\beq \no
\mathcal{J}\widetilde{Q}_a'(t)=a H(t)\widetilde{Q}_a(t) \quad \widetilde{Q}_a(0)=Id.
\eeq
Moreover, it is clear by linearity that the sequence $\{n_k\}$ is independent of $a$ and that $H_a(t)=aH_1(t)\equiv a H(t)$. Also $H(t) \geq 0$ for almost every $t\in[0,1]$, since $H_{x_0,j} \geq 0$ for every $j$. Finally, the convergence is uniform for $a$ in compact sets in $\mathbb{C}$.

Now, by \eqref{eq:ScaledCDKernel}, we see that for any $a \neq b$
\beq \no
\begin{split}
&\lim_{k \rightarrow \infty}\frac{1}{n_k}K_{n_k}\left(x_0+\frac{a}{n_k}, x_0+\frac{b}{n_k} \right) \\
& \quad =\frac{1}{a-b} \det \left( \widetilde{Q}_a(1) \left(\begin{array}{c} 1 \\ 0 \end{array} \right), \widetilde{Q}_b(1) \left( \begin{array}{c} 1 \\ 0 \end{array} \right)\right).
\end{split}
\eeq
By \eqref{eq:ReproducingDeterminant}
\beq \no
\frac{1}{a-b} \det \left( \widetilde{Q}_a(1) \left(\begin{array}{c} 1 \\ 0 \end{array} \right), \widetilde{Q}_b(1) \left( \begin{array}{c} 1 \\ 0 \end{array} \right)\right)=K_H(\overline{a},b)
\eeq
where $K_H(z,\zeta)$ is the reproducing kernel of the de Branges space associated with the canonical system defined by $H$.

To deal with the case $a=b$ fix $a \in \mathbb{C}$ and note that for any $k$,
\beq \no
\frac{1}{n_k}K_{n_k}\left(x_0+\frac{a}{n_k}, x_0+\frac{b}{n_k} \right)
\eeq
is entire as a function of $b$. In addition, $K_H(a,b)$ is an entire function of $b$ (as the reproducing kernel of a de Branges space), Thus, by the uniform convergence in \eqref{eq:ConvergenceUniform1}, the convergence is uniform on annuli around $a$, which implies (using Cauchy's formula) that
\beq \no
\begin{split}
&\frac{1}{n_k}K_{n_k}\left(x_0+\frac{a}{n_k}, x_0+\frac{a}{n_k} \right) \\
&\quad =\frac{1}{2\pi i}\int_{|b-a|=1}\frac{1}{n_k}K_{n_k}\left(x_0+\frac{a}{n_k}, x_0+\frac{b}{n_k} \right)\frac{\textrm{d}b}{b-a}
\end{split}
\eeq
converges to
\beq \no
\frac{1}{2\pi i}\int_{|b-a|=1}K_H(\overline{a},b)\frac{\textrm{d}b}{b-a}=K_H(\overline{a},a).
\eeq
This finishes the proof.
\end{proof}

\begin{proof}[Proof of Theorem \ref{thm:Sine}]
Assume first that (i) holds. Note that, in particular, this means that
\beq \no
\frac{1}{n}\sum_{j=0}^{n-1}p_j\left(x_0 \right)^2 \underset{n \rightarrow \infty}{\longrightarrow} \frac{\rho \left(x_0 \right)}{w\left( x_0\right)}.
\eeq
Furthermore, by \cite[Theorem 1.3]{bls}, (i) implies that
\beq \no
\frac{1}{n}\sum_{j=0}^{n-1}q_j\left(x_0 \right)^2 \underset{n \rightarrow \infty}{\longrightarrow}\frac{\rho \left(x_0 \right)}{\widetilde{w}\left( x_0\right)},
\eeq
and, in addition, by the proof of \cite[Theorem 1.5]{bls}, we also have that
\beq \no
\frac{1}{n}\sum_{j=0}^{n-1}p_j\left(x_0\right)q_j \left(x_0 \right) \underset{n \rightarrow \infty}{\longrightarrow} \textrm{Re}F\left(x+i0 \right)\frac{\rho \left(x_0 \right)}{w\left( x_0\right)}.
\eeq
Thus, for $H_{x_0,j}$ as defined in \eqref{eq:H}, we have that
\beq \no
\lim_{n \rightarrow \infty}\frac{1}{n}\sum_{j=0}^{n-1} H_{x_0,j}=H
\eeq
where
\beq \label{eq:HIs}
H\equiv \left( \begin{array}{cc} \frac{\rho(x_0)}{w(x_0)} & -\textrm{Re}F(x_0+i0)\frac{\rho(x_0)}{w(x_0)} \\
-\textrm{Re}F(x_0+i0)\frac{\rho(x_0)}{w(x_0)} & \frac{\rho(x_0)}{\widetilde{w}(x_0)} \end{array} \right).
\eeq
By Lemma \ref{lem:NiceConvergence} (and by linearity in $a$) this implies that uniformly for $a$ in compacts and $t \in [0,1]$
$$
\lim_{n\rightarrow \infty} Q_{a/n}([tn])=\widetilde{Q}_a(t)
$$
where $\widetilde{Q}_a(t)$ solves
\beq \label{eq:QSolves}
\left( \begin{array}{cc} 0 & -1 \\
1 & 0 \end{array} \right) \widetilde{Q}'(t)=aH\widetilde{Q}(t)  \quad \widetilde{Q}(0)=Id.
\eeq

\smallskip

Now assume that (ii) holds. By \eqref{eq:ScaledCDKernel}, this implies that for $a \neq b$
\beq \no
\lim_{n \rightarrow \infty}\frac{K_n\left(x_0+\frac{a}{n},x_0+\frac{b}{n} \right)}{n}=\frac{1}{a-b} \det \left( \widetilde{Q}_a(1) \left(\begin{array}{c} 1 \\ 0 \end{array} \right), \widetilde{Q}_b(1) \left( \begin{array}{c} 1 \\ 0 \end{array} \right)\right)
\eeq
where $\widetilde{Q}$ solves \eqref{eq:QSolves} and $H$ is \eqref{eq:HIs}. By
\beq \no
\widetilde{w}\left(x_0 \right)=\frac{w\left(x_0\right)}{\left|F\left(x_0+i0\right) \right|^2}=\frac{w\left(x_0\right)}{\pi^2 w\left(x_0 \right)^2+\textrm{Re}F\left(x_0+i0\right)^2}
\eeq
(see, e.g., \cite[(2.13)]{bls}) we see that $\det H=\pi^2 \rho \left(x_0 \right)^2>0$. By integrating \eqref{eq:QSolves} (this is an ODE with constant coefficients) we get that
\beq \no
\begin{split}
\frac{1}{a-b} \det \left( \widetilde{Q}_a(1) \left(\begin{array}{c} 1 \\ 0 \end{array} \right), \widetilde{Q}_b(1) \left( \begin{array}{c} 1 \\ 0 \end{array} \right)\right)&=\frac{\sin\left(\sqrt{\det H}(b-a) \right)}{\frac{w\left(x_0 \right)}{\rho\left(x_0 \right)}\sqrt{\det H}(b-a)}\\
&=\frac{\sin\left(\pi \rho \left(x_0 \right)(b-a) \right)}{\pi w\left(x_0 \right)(b-a)}.
\end{split}
\eeq
Since all functions involved are entire and the convergence is uniform on compact sets, we may use Cauchy's formula as in the proof of Theorem \ref{thm:Compactness} to show convergence for $a=b$. This finishes the proof.
\end{proof}

\section{An example}

In this section we sketch the details of an example illustrating the power of the approach demonstrated here, to deal with perturbations of order $O(1/n)$. Let $J_n$ be the Jacobi matrix with parameters defined by
\beq \no
a_{n,j}\equiv 1 \quad b_{n,j}=\frac{(-1)^{j+1} V}{n}
\eeq
for some $V>0$. For  $x_0=0$ we want to compute the scaling limit of the difference equation (i.e.\ the corresponding canonical system) and the scaling limit of the CD kernel:
\beq \label{eq:LimitCD}
\lim_{n \rightarrow \infty}\frac{K^{(V)}_n\left(a/n,b/n \right)}{n}.
\eeq
Note that since the Jacobi coefficients themselves depend on $n$, this is somewhat different than the theorems considered in the previous sections. Nevertheless, we believe this example can be instructive.

We first want to compute the transfer matrices $T_\ell^{(n)}(0)$. We do that by applying a similar idea to the one applied in Section 2. Namely, we let $T_\ell^{(0)}(0)$ be the transfer matrices at $0$ for $V=0$ and consider the recurrence equation for
\beq \label{eq:Qhat}
\widehat{Q}_\ell^{(n)}(0)=\left(T_\ell^{(0)}(0)\right)^{-1}T_\ell^{(n)}(0).
\eeq
Note that for all $\ell$
\beq \no
S_\ell^{(0)}(0)=\left(\begin{array}{cc} 0 & -1 \\ 1 & 0 \end{array} \right)
\eeq
so we get the following formulas for $T_\ell^{(0)}(0)$: 
\beq \no
T_\ell^{(0)}(0)=\left(\begin{array}{cc} 0 &-1 \\ 1 & 0 \end{array} \right)
\eeq
for $\ell=1 \mod 4$,
\beq \no
T_\ell^{(0)}(0)=\left(\begin{array}{cc} -1 & 0 \\ 0 & -1 \end{array} \right)
\eeq
for $\ell =2 \mod 4$,
\beq \no
T_\ell^{(0)}(0)=\left(\begin{array}{cc} 0 &1 \\ -1 & 0 \end{array} \right)
\eeq
for $\ell =3 \mod 4$, and 
\beq \no
T_\ell^{(0)}(0)=\left(\begin{array}{cc} 1 & 0 \\ 0 & 1 \end{array} \right)
\eeq
for $\ell =0 \mod 4$.

It follows that for $(\ell+1)$ even
\beq \no
\widehat{Q}_{\ell+1}^{(n)}(0)=\left(\left( \begin{array}{cc} 1 & V/n \\ 0 & 1 \end{array} \right)\left( \begin{array}{cc} 1 & 0 \\ V/n & 1 \end{array} \right) \right)^{(\ell+1)/2}
\eeq
and for $(\ell+1)$ odd
\beq \no
\widehat{Q}_{\ell+1}^{(n)}(0)=\left(\begin{array}{cc} 1 & 0 \\ V/n & 1 \end{array} \right)\left(\left( \begin{array}{cc} 1 & V/n \\ 0 & 1 \end{array} \right)\left( \begin{array}{cc} 1 & 0 \\ V/n & 1 \end{array} \right) \right)^{\ell/2}.
\eeq

Denoting
\beq \no
\lambda_\pm=1+\frac{V^2}{2n^2}\pm\frac{V}{n}\sqrt{1+\frac{V^2}{4n^2}}.
\eeq
and
\beq \no
U_n=\left(\begin{array}{cc} 1 & 1 \\ -\sqrt{1+\frac{V^2}{4n^2}}-\frac{V}{2n} & \sqrt{1+\frac{V^2}{4n^2}}-\frac{V}{2n} \end{array} \right),
\eeq
we get that
\beq \no
\left(\left( \begin{array}{cc} 1 & V/n \\ 0 & 1 \end{array} \right)\left( \begin{array}{cc} 1 & 0 \\ V/n & 1 \end{array} \right) \right)
=U_n\left(\begin{array}{cc} \lambda_- &0 \\ 0 & \lambda_+ \end{array} \right)U_n^{-1}.
\eeq
It follows that for $(\ell+1)$ even
\beq \label{eq:evenDiagonal}
\widehat{Q}_{\ell+1}^{(n)}(0)=U_n\left(\begin{array}{cc} \lambda_-^{\frac{\ell+1}{2}} &0 \\ 0 & \lambda_+^{\frac{\ell+1}{2}} \end{array} \right)U_n^{-1}
\eeq
and for $(\ell+1)$ odd
\beq \label{eq:oddDiagonal}
\widehat{Q}_{\ell+1}^{(n)}(0)=\left(\begin{array}{cc} 1 & 0 \\ V/n & 1 \end{array} \right)U_n\left(\begin{array}{cc} \lambda_-^{\frac{\ell}{2}} &0 \\ 0 & \lambda_+^{\frac{\ell}{2}} \end{array} \right)U_n^{-1},
\eeq
and we may combine these formulas with \eqref{eq:Qhat} and the formulas for $T_\ell^{(0)}(0)$ to obtain an expression for $T_\ell^{(n)}(0)$ (which we omit here).

We now want to consider the difference equation for
\beq \no
Q_\ell^{(n)}(a/n)=\left(T_\ell^{(n)}(0) \right)^{-1} T_\ell^{(n)}(a/n).
\eeq
Because of the $U_n^{-1}$ in \eqref{eq:evenDiagonal} and \eqref{eq:oddDiagonal}, it is in fact simpler to first deal with the difference equation for
\beq \no
\mathcal{Q}_\ell^{(n)}(a/n)=U_n^{-1}Q_\ell^{(n)}(a/n).
\eeq
Combining \eqref{eq:Qhat} with \eqref{eq:evenDiagonal}, \eqref{eq:oddDiagonal} and with the second equality in \eqref{eq:VariationOfParameters} we get that
\beq \label{eq:ExampleDifference}
\mathcal{Q}_{\ell+1}^{(n)}(a/n)-\mathcal{Q}_{\ell}^{(n)}(a/n)=-\frac{a}{n}A_\ell^{(n)}\mathcal{Q}_\ell^{(n)}(a/n)
\eeq
where, for $\ell$ even
\beq \label{eq:Aelleven}
A_\ell^{(n)}=\frac{1}{2\sqrt{1+\frac{V^2}{4n^2}}}\left(\begin{array}{cc} -1 & -\lambda_+^\ell \\ \lambda_-^\ell &1 \end{array} \right)
\eeq
and for $\ell$ odd
\beq \label{eq:Aelleven}
A_\ell^{(n)}=\frac{1}{2\sqrt{1+\frac{V^2}{4n^2}}}\left(\begin{array}{cc} 1 & -\lambda_+^{\ell-2} \\ \lambda_-^{\ell-2} & -1 \end{array} \right)+O\left( \frac{V}{n}\right),
\eeq
(the $O\left(\frac{V}{n} \right)$ is an $\ell,n$-dependent matrix whose norm is bounded, uniformly in $\ell$, by $\frac{V}{n}$).

It thus follows that
\beq \no
\sup_{\ell \leq n,n}\left \|A_\ell^{(n)} \right \|<\infty
\eeq
so that, in particular, \eqref{eq:AverageBound}, \eqref{eq:LimitBound} and \eqref{eq:AverageDecay} are satisfied. Moreover, it is not hard to see that, for $\widetilde{A}^{(n)}(t)$ defined from $A_\ell^{(n)}$ as in Section 3, (note the alternating signs on the diagonals and the leading behavior $\lambda_{\pm} \sim 1\pm \frac{V}{n}$)
\beq \no
\sup_{0 \leq t \leq 1} \left \|\int_{0}^t \left(\widetilde{A}^{\left(n \right)}(s)-A(s) \right)\textrm{d}s \right \| \rightarrow 0
\eeq
where
\beq \no
A(s)=\left(\begin{array}{cc} 0 & -\frac{e^{sV}}{2} \\ \frac{e^{-sV}}{2} & 0 \end{array} \right).
\eeq
It therefore follows from Lemma \ref{lem:ContinuumConvergence} that $\widetilde{\mathcal{Q}}^{(n)}$ converges, uniformly in $t \in [0,1]$, to the solution of
\beq \label{eq:DifferentialSolvable}
\mathcal{Q}'(t)=-a A(t) \mathcal{Q}(t)
\eeq
with
\beq \no
\mathcal{Q}(0)=\lim_{n \rightarrow \infty} U_n^{-1}=\left( \begin{array}{cc} \frac{1}{2} & -\frac{1}{2} \\ \frac{1}{2} & \frac{1}{2} \end{array} \right)=U_\infty^{-1},
\eeq
(Lemma \eqref{lem:ContinuumConvergence} is phrased for fixed boundary conditions, but the proof is easily modifiable to converging boundary conditions).

By conjugating with $U_\infty=\left( \begin{array}{cc} 1 & 1 \\ -1 & 1 \end{array} \right)$ and $U_\infty^{-1}$, and by multiplying with $\mathcal{J}$, it now follows that $Q_\ell^{(n)}(a/n)$ converges in the sense described in Section 3 to a solution of the canonical system
\beq \label{eq:CanonicalSinh}
\mathcal{J}Q_a'(t)=a \left(\begin{array}{cc} \frac{\cosh(tV)}{2} & \frac{\sinh(tV)}{2} \\ \frac{\sinh(tV)}{2} & \frac{\cosh(tV)}{2} \end{array} \right)Q_a(t).
\eeq
Note that for $V=0$ we indeed get the limit for the Jacobi matrix with $a_n \equiv 1, \ b_n \equiv 0$ at $x_0=0$.

Finally, in order to compute \eqref{eq:LimitCD} we note that \eqref{eq:DifferentialSolvable} can be transformed into a constant coefficient second order ODE (by differentiating twice any of the entries). By integrating the resulting equation and performing the necessary transformations in order to compute
\beq \no
\det \left(Q_a(1) \left(\begin{array}{c} 1 \\ 0 \end{array} \right), Q_b(1) \left(\begin{array}{c} 1 \\ 0 \end{array} \right) \right),
\eeq
we get that for all $a,b \in \mathbb{C}$ (uniformly on compacts)
\beq \label{eq:ModifiedSine}
\begin{split}
&\lim_{n \rightarrow \infty} \frac{K^{(V)}_n\left(\frac{a}{n},\frac{b}{n} \right)}{n} \\
&\quad= a\frac{\sinh\left(\omega_a/2 \right)\cosh\left(\omega_b/2 \right)}{\omega_a}-b\frac{\sinh\left(\omega_b/2 \right)\cosh\left(\omega_a/2 \right)}{\omega_b}\\
&\quad+V\left(\frac{\sinh\left(\omega_b/2 \right)\cosh\left(\omega_a/2 \right)}{\omega_b}- \frac{\sinh\left(\omega_a/2 \right)\cosh\left(\omega_b/2 \right)}{\omega_a}\right)
\end{split}
\eeq
where $\omega_x=\sqrt{V^2-x^2}$.
\begin{remark}
Note that the function on the right hand side of \eqref{eq:ModifiedSine} is entire. This can be seen by writing the power series of the hyperbolic functions and noting that $\omega_a$ and $\omega_b$ appear throughout only with even powers. 
\end{remark}
\begin{remark}
It is an interesting exercise to verify that indeed, when $V=0$ one gets
\beq \no
\frac{\sin\left((a-b)/2 \right)}{a-b}
\eeq
as expected.
\end{remark}
\section{Appendix: Canonical systems, Jacobi matrices, and the CD kernel}

Canonical systems generalize many second order difference and differential operators, with Schr\"odinger, Dirac and Jacobi being particular cases. Furthermore, the correspondence between such systems and Hermite-Biehler functions (see below for a definition), established by de Branges \cite{debranges} is a central result in the theory of de Branges spaces. Thus there is a huge literature on canonical systems, spanning spectral theory, harmonic analysis and number theory (\cite{arov-dym,debranges,dym,Hassi,Lagarias,Remling,Romanov} are only a few relevant references). The next several paragraphs contain only a quick a review of some results that are directly relevant here (for details, see e.g., \cite{Remling}).

A canonical system is a family (indexed by $z \in \mathbb{C}$) of differential equations of the form
\beq \label{eq:CanonicalSystem}
\mathcal{J}u'(t)=z H(t)u(t)
\eeq
on some interval $I=[0,L] \subseteq \mathbb{R}$, where $\mathcal{J}=\left( \begin{array}{cc} 0 & -1 \\
1 & 0 \end{array} \right)$ and $H(t)$ is a $2\times2$ nonnegative definite matrix valued function such that the entries of $H$ are integrable functions on $I$. By a change of variable (see, e.g., \cite[Section 6]{Remling} or \cite{Romanov}) we may also assume that $H(t) \not \equiv 0$ on nonempty open sets.

In the case that $H(t)$ is invertible almost everywhere, we may rewrite \eqref{eq:CanonicalSystem} as
\beq \no
H^{-1}(t)\mathcal{J} u'(t)=z u(t)
\eeq
i.e., as an eigenvalue equation for the operator $H^{-1}(t)\mathcal{J}\frac{d}{dt}$ which is symmetric with respect to the inner product
\beq \label{eq:IPCanonical}
\left( f,g \right)_H=\int_0^L \left(f(t),H(t) g(t)\right)_{\mathbb{C}^2}\textrm{d}t.
\eeq
Let $\mathcal{H}_H$ be the Hilbert space of vector valued functions on $I$ corresponding to this inner product. Choosing appropriate boundary conditions at $0$ and $L$ (e.g., $f(0)=f(L)=\left(\begin{array}{c} 1 \\ 0 \end{array} \right)$) defines a domain of self-adjointness for this operator. Moreover, even if $H$ is not invertible a.e., as long as the boundary condition at $0$ is not orthogonal to $\textrm{Image}(H)(t)$ $\forall t\in(0,\varepsilon)$ for some $\varepsilon>0$ (and a similar condition is satisfied at $L$), it is possible to define a subspace of $\mathcal{H}_H$ such that \eqref{eq:CanonicalSystem} is the eigenvalue equation for a self-adjoint operator defined on that subspace (see \cite[Section 2]{Romanov}).

A solution to \eqref{eq:CanonicalSystem} is an absolutely continuous function, $u: [0,L]\rightarrow \mathbb{C}^2$, that satisfies \eqref{eq:CanonicalSystem} a.e.

Let $u(t,z)$ be a solution with initial value $u(0)=\left(\begin{array}{c} 1 \\ 0 \end{array} \right)$. Then the function $E_L(z)=u_1(L,z)+iu_2(L,z)$ is a \emph{Hermite-Biehler} function, i.e.\ it has no zeros in the upper half plane $\mathbb{C}^+=\left\{z \mid \textrm{Im}z>0 \right\}$ and satisfies
\beq \no
\left|E_L(z) \right| \geq \left|E_L(\overline{z}) \right| \textrm{ for } z \in \mathbb{C}^+.
\eeq
Such functions are at the basis of the theory of de Branges spaces.

The de Branges space $B(E)$associated with a Hermite-Biehler function $E$ is the set of all entire functions, $f$, such that both $\frac{f}{E}$ and $\frac{f^\sharp}{E}$ are in $H^2(\mathbb{C}^++)$, where $f^\sharp(z)=\overline{f(\overline{z})}$. It is a reproducing kernel Hilbert space with inner product given by
\beq \no
\left(f,g \right)_E=\frac{1}{\pi}\int_\mathbb{R} \overline{f(x)}g(x)\frac{\textrm{d}x}{|E(x)|^2}.
\eeq
and reproducing kernel
\beq \no
K_E(z,\zeta)=\frac{\overline{E(z)}E(\zeta)-E(\overline{z})\overline{E(\overline{\zeta})}}{2i(\overline{z}-\zeta)}.
\eeq

It follows that every canonical system \eqref{eq:CanonicalSystem} gives rise to a de Branges space through the function $E_L(z)$ associated with its solution as described above. In fact, in this particular case, it is not hard to show that the reproducing kernel is also given by
\beq \label{eq:ReproducingDeterminant}
K_{E_L}(z,\zeta)=\frac{1}{\overline{z}-\zeta}\det\left(\begin{array}{cc} u_1(L,\overline{z}) & u_1(L,\zeta) \\ u_2(L,\overline{z}) & u_2(L,\zeta) \end{array} \right),
\eeq
and by
\beq \label{eq:ReproducingDeBranges}
K_{E_L}(z,\zeta)=\left( u(\cdot,z),u(\cdot,\zeta) \right)_H=\int_0^L \left(u(t,z),H(t) u(t,\zeta)\right)_{\mathbb{C}^2}\textrm{d}t.
\eeq
In particular, for any fixed $z$, $K_{E_L}(z,\zeta)$ is an entire function of $\zeta$.

It turns out that this is not a particular example, but rather the general case: a fundamental result in the theory of de Branges spaces (see, e.g.\ \cite{debranges,Remling,Romanov}) says that every de Branges space is associated with a canonical system.

As shown in \cite[Section 4]{LubJFA}, the CD kernel $K_n(x,y)$ associated with a measure $\mu$ is a reproducing kernel for a de Branges space as well. Indeed, let
\beq \no
L_n(x,y)=(x-y)K_n(x,y).
\eeq
Then the Christoffel-Darboux formula \eqref{eq:CDFormula} says that
\beq \label{eq:CDFormula1}
L_n(x,y)=a_n\left(p_n(x)p_{n-1}(y)-p_n(y)p_{n-1}(x) \right).
\eeq

Now for any fixed $w \in \mathbb{C}^+$, the function
\beq \no
E_{n,w}(\cdot)=\sqrt{2}\frac{L_n(\overline{w},\cdot)}{\left|L_n(w,\overline{w}) \right|^{1/2}}
\eeq
is a Hermite-Biehler function. The corresponding de Branges space, $B(E_{n,w})$, is the space of polynomials of degree$<n$ and its reproducing kernel is
\beq \no
K_{E_{n,w}}(z,\zeta)=K_n(\overline{z},\zeta)
\eeq
(as can be seen easily from \eqref{eq:CDFormula}). Note that the definition in \cite[Theorem 4.3]{LubJFA} differs from ours by a factor of $\sqrt{\pi}$; this is because we have an extra factor of $\pi$ in the inner product defining $B(E)$. As shown in Section 2 above, this de Branges space is naturally associated with the discrete canonical system \eqref{eq:DiscreteCanonical} with $x_0=0$.

It is an interesting fact that it is possible to also go in the other direction and associate a Jacobi matrix with any discrete canonical system satisfying the appropriate conditions. Let $\{r_\ell\}_{\ell=0}^\infty$ and $\{s_\ell\}_{\ell=0}^\infty$ be two real sequences satisfying
\beq \no
s_\ell r_{\ell-1}-r_{\ell}s_{\ell-1}=\frac{1}{a_\ell}
\eeq
for some sequence of positive numbers $\{a_\ell\}_{\ell=0}^\infty$ with $a_0=1$. Consider the discrete canonical system
\beq \no
\begin{split}
&\mathcal{J}\left(\widehat{u}_{\ell+1}(x)-\widehat{u}_\ell(x)\right) \\
&\quad=x\left(\begin{array}{cc} r_\ell^2 & -s_\ell r_\ell \\ -s_\ell r_\ell & s_\ell^2 \end{array} \right) \widehat{u}_\ell(x).
\end{split}
\eeq
Then if $\{u_\ell(x)\}$ is a solution satisfying with $u_0=\left(\begin{array}{c} 1 \\ 0 \end{array} \right)$ then
\beq \no
p_\ell(x)=r_\ell u_{\ell,1}(x)-s_\ell u_{\ell,2}(x)
\eeq
is the $\ell$'th orthonormal polynomial with respect to the Jacobi matrix whose off diagonal parameter sequence is the given sequence $\{a_\ell\}_{\ell=1}^\infty$, and whose diagonal entries are
\beq \no
b_\ell=a_\ell a_{\ell-1}\left(r_\ell s_{\ell-2}-s_\ell r_{\ell-2} \right).
\eeq
That this is true follows by a direct computation writing
\beq \no
\widehat{T}_{\ell}(0)=\left( \begin{array}{cc} a_\ell r_\ell & -a_\ell s_\ell \\ r_{\ell-1} & -s_{\ell-1} \end{array} \right)
\eeq
and noting that $\widehat{T}_\ell(0)\widehat{T}_{\ell-1}^{-1}(0)=S_\ell(0)$ from which one may compute the values of the Jacobi parameters. For a similar analysis in the continuum, associating canonical systems to one-dimensional Schr\"odinger operators, see \cite[Section 8]{Remling}.


\begin{thebibliography}{cc}

\bibitem{arov-dym} D.~Z.~Arov and H.~Dym, {\it J-inner matrix functions, interpolation and inverse problems for canonical systems, I: Foundations}, Integr.\ Equ.\ Oper.\ Theory, \textbf{29} (1997), 373--454.

\bibitem{avila} A.~Avila, {\it On the Kotani-Last and Schr\"odinger conjectures}, J.\ Amer.\ Math.\ Soc.\ \textbf{28} (2015), 579-–616.

\bibitem{als} A.~Avila, Y.~Last and B.~Simon, {\it Bulk universality and clock spacing of zeros for ergodic Jacobi matrices with a.c.\ spectrum}, Analysis and PDE, {\bf 3} (2010), 81--108.


\bibitem{breuerSing} J.~Breuer, {\it Sine kernel asymptotics for a class of singular measures}, J.\ Approx.\ Theory \textbf{163} (2011), 1478–-1491.

\bibitem{BL} J.~Breuer and Y.~Last, {\it Stability of spectral types for Jacobi matrices under decaying random perturbations},  J.\ Funct.\ Anal.\ \textbf{245} (2007), 249–-283.

\bibitem{bls} J.~Breuer, Y.~Last and B.~Simon, {\it Stability of asymptotics of Christoffel-Darboux kernels}, Commun.\ Math.\ Phys.\ \textbf{330} (2014), 1155–-1178.


\bibitem{deift} P.~Deift, {\it Orthogonal Polynomials and Random Matrices: A Riemann-Hilbert Approach}, Courant Institute Lecture Notes, {\bf 3} New York University Press, New York, 1999.

\bibitem{debranges} L.~de Branges, {\it Hilbert spaces of entire functions}, Prentice-Hall, NJ (1968).

\bibitem{dym} H.~Dym, {\it An introduction to de Branges spaces of entire functions with applications to differential equations of the Sturm-Liouville type}, Adv.\ Math.\ \textbf{5} (1971), 395--471.


\bibitem{EdelSutt} A.~Edelman and B.~D.~Sutton, {\it From random matrices to stochastic operators}, J.\ Stat.\ Phys.\ \textbf{127} (2007), 1121--1165.

\bibitem{ET} P.~Erd\"os and P.~Tur\'an, {\it On interpolation. III. Interpolatory theory of polynomials}, Annals of Math.\ (2) \textbf{41} 1940, 510--553.

\bibitem{findley} E.~Findley, {\it Universality for regular measures satisfying Szeg\H o's condition}, J.\ Approx.\ Theory, {\bf 155} (2008), 136--154.

\bibitem{Freud} G.~Freud,
\textit{Orthogonal Polynomials},
Pergamon Press, Oxford-New York, 1971.


\bibitem{Hassi} S.~Hassi, H.~De Snoo and H.~Winkler, {\it Boundary-value problems for two-dimensional canonical systems}, Integr.\ Equ.\ Oper.\ Theory, \textbf{36} (2000), 445--479.

\bibitem{Konig} W.~K\"onig, {\it Orthogonal polynomial ensembles in probability theory}, Probab.\ Surveys \textbf{2} (2005), 385-–447.

\bibitem{KRV} M.~Krishnapur, B.~Rider and B.~Vir\'ag, {\it Universality of the stochastic Airy operator},  Comm.\ Pure Appl.\ Math.\ \textbf{69} (2016), 145–-199.


\bibitem{KVV} E.~Kritchevski, B.~Valk\'o and B.~Vir\'ag, {\it The scaling limit of the critical one-dimensional random Schr\"oinger operator}, Commun.\ Math.\ Phys.\ \textbf{314} (2012), 775--806.


\bibitem{KV} A.~B.~J.~Kuijlaars and M~Vanlessen, {\it Universality for eigenvalue correlations from the modified Jacobi unitary ensemble}, Int.\ Math.\ Res.\ Not.\ 2002, 1575–-1600.

\bibitem{Lagarias} J.~C.~Lagarias, {\it Hilbert spaces of entire functions and Dirichlet L-functions}, Frontiers in number theory, physics, and geometry. I, Springer, Berlin, 365--377.

\bibitem{Last-Simon} Y.~Last and B.~Simon, {\it Fine structure of the zeros of orthogonal polynomials, IV. A priori bounds and clock behavior}, Comm.\ Pure Appl.\ Math.\ \textbf{61} (2008), 486--538.


\bibitem{LeLu} E.~Levin, and D.~Lubinsky {\it Application of universality limits to zeros and reproducing kernels of
orthogonal polynomials}, J.\ Approx.\ Theory {\bf 150} (2008), 69--95.

\bibitem{lubinsky2} D.~S.~Lubinsky, {\it Universality limits in the bulk for arbitrary measures on compact sets}, J.\ Anal.\ Math., {\bf 106} (2008), 373--394.

\bibitem{lubinsky1} D.~S.~Lubinsky, {\it A new approach to universality involving orthogonal polynomials}, Annals of Math., {\bf 170} (2009), 915--939.


\bibitem{lubinskyRev} D.~S.~ Lubinsky, {\it Some recent methods for establishing universality limits}, J.\ Nonlinear Anal., {\bf 71} (2009), e2750--e2765.


\bibitem{LubJFA} D.~S.~Lubinsky, {\it Universality limits for random matrices and de Branges spaces of entire functions}, J.\ Funct.\ Anal.\ \textbf{256} (2009), 3688--3729.

\bibitem{lubinskyRev2} D.~S.~Lubinsky, {\it An update on local universality limits for correlation functions generated by unitary ensembles}, SIGMA \textbf{12} (2016), 36 pp.


\bibitem{RRV} J.~A.~Ram\'irez, B.~Rider and B.~Vir\'ag, {\it Beta ensembles, stochastic Airy spectrum, and a diffusion},  J.\ Amer.\ Math.\ Soc.\ \textbf{24} (2011), 919-–944.

\bibitem{Remling} C.~Remling, {\it Schr\"odinger operators and de Branges spaces}, J.\ Funct.\ Anal.\ \textbf{196} (2002), 323-–394.

\bibitem{Romanov} R.~Romanov, {\it Canonical systems and de Branges spaces}, Preprint, \texttt{arXiv} math.SP 1408.6022.

\bibitem{simonBounded} B.~Simon, {\it Bounded eigenfunctions and absolutely continuous spectra for one-dimensional Schr\"odinger operators}, Proc.\ Amer.\ Math.\ Soc.\ \textbf{124} (1996), 3361--3369.

\bibitem{SimonZeros} B.~Simon, {\it Fine structure of the zeros of orthogonal polynomials: A review}, Difference Equations, Special Functions and Orthogonal Polynomials (eds.\ S.~Elaydi et al.), World Sci.\ Publ., Singapore, (2007) 636--653.

\bibitem{Simon-ext} B.~Simon, {Two extensions of Lubinsky's universality theorem}, J.\ Anal.\ Math., {\bf 105} (2008), 345--362.

\bibitem{Simon-CD} B.~Simon, {\it The Christoffel-Darboux kernel}, in ``Perspectives in PDE, Harmonic Analysis and Applications'', pp 295--335, Proc.\ Sympos.\ Pure Math.\ {\bf 79}, American Mathematical Society, Providence, RI, 2008.


\bibitem{SimonSz} B.~Simon,
{\it Szeg\H o's Theorem and its Descendants},
M.~B.~Porter Lectures, Princeton University Press, Princeton, NJ, 2011.


\bibitem{totik} V.~Totik, {\it Universality and fine zero spacing on general sets}, Ark.\ Mat., {\bf 47} (2009), 361--391.



\bibitem{VV1} B.~Valk\'o and B.~Vir\'ag, {\it Continuum limits of random matrices and the Brownian carousel}, Invent.\ Math.\ \textbf{177} (2009), 463--508.

\bibitem{VV2} B.~Valk\'o and B.~Vir\'ag, {\it The $\sin_\beta$ operator}, Invent.\ Math.\ \textbf{209} (2017), 275--327.


\bibitem{ViragRev} B.~Vir\'ag, {\it Operator limits of random matrices}, Proceedings of the International Congress of Mathematicians—Seoul 2014, \textbf{IV}, 247-–271.

\bibitem{XZZ} S.~X.~Xu, Y.~Q.~Zhao and J.~R.~Zhou {\it Universality for eigenvalue correlations from the unitary ensemble associated with a family of singular weights}, J.\ Math.\ Phys.\ \textbf{52} (2011), 14 pp.

\end{thebibliography}
\end{document}